\documentclass[11pt]{article}

\usepackage{a4}
\usepackage{authblk}
\usepackage{amsmath,amssymb}
\usepackage{amsthm}
\usepackage[dvipdfmx]{graphicx}

\usepackage{ascmac}
\usepackage{tabls}
\usepackage{booktabs}
\usepackage[dvipdfmx]{graphicx}
\usepackage{wrapfig}
\usepackage{array}
\usepackage[algoruled,linesnumbered,algo2e,vlined]{algorithm2e}
\usepackage{url}
\usepackage{multirow}
\usepackage{multicol}
\usepackage{cases}
\usepackage[normalem]{ulem}
\usepackage{boites}
\usepackage[final,mode=multiuser]{fixme}

\usepackage{colortbl}

\fxsetup{theme=color}
\definecolor{fxtarget}{rgb}{0.0000,0.0000,0.4823}

\fxuseenvlayout{colorsig}
\fxusetargetlayout{color}

\fxuseenvlayout{colorsig}
\fxusetargetlayout{color}
\FXRegisterAuthor{ko}{ako}{KO}
\FXRegisterAuthor{tm}{atm}{TM}
\FXRegisterAuthor{yn}{ayn}{YN}
\FXRegisterAuthor{si}{asi}{SI}
\FXRegisterAuthor{hb}{ahb}{HB}
\fxsetface{env}{}

\pdfoutput=1

\usepackage{thmtools}

\newtheorem{theorem}{Theorem}
\newtheorem{lemma}{Lemma}
\newtheorem{corollary}{Corollary}

\newtheorem{definition}{Definition}
\newtheorem{problem}{Problem}

\newtheorem{example}{Example}

\renewenvironment{proof}{\begin{trivlist} \item{\textit{Proof.}}}{\end{trivlist}}

\newcommand{\AW}{\mathsf{AW}}
\newcommand{\MAW}{\mathsf{MAW}}
\newcommand{\nonMAW}{\mathsf{non\-MAW}}

\newcommand{\Substr}{\mathsf{Substr}}

\newcommand{\DAWG}{\mathsf{DAWG}}
\newcommand{\EndPos}{\mathsf{End\_Pos}}
\newcommand{\Eqc}[1]{[{#1}]}

\newcommand{\B}{\mathbf{B}}
\newcommand{\Label}{\mathsf{label}}
\newcommand{\Schar}{\mathsf{schar}}
\newcommand{\Char}{\mathsf{char}}

\newcommand{\Succ}{\mathsf{succ}}

\begin{document}

\title{Linear-time computation of \\ generalized minimal absent words for multiple strings}

\author[1]{Kouta~Okabe}
\author[2]{Takuya~Mieno}
\author[3]{Yuto~Nakashima}
\author[3]{Shunsuke~Inenaga}
\author[4]{Hideo~Bannai}

\affil[1]{Department of Information Science and Technology, Kyushu University}

\affil[2]{Department of Computer and Network Engineering, University of Electro-Communications}

\affil[3]{Department of Informatics, Kyushu University}

\affil[4]{M\&D Data Science Center, Tokyo Medical and Dental University}

\date{}
\maketitle

\begin{abstract}
A string $w$ is called a \emph{minimal absent word} (\emph{MAW}) for a string $S$ if $w$ does not occur as a substring in $S$ and all proper substrings of $w$ occur in $S$.
MAWs are well-studied combinatorial string objects that have potential applications in areas including bioinformatics, musicology, and data compression.
In this paper, we generalize the notion of MAWs to a set $\mathcal{S} = \{S_1, \ldots, S_k\}$ of multiple strings.
We first describe our solution to the case of $k = 2$ strings,
and show how to compute the set $\mathsf{M}$ of MAWs in optimal $O(n + |\mathsf{M}|)$ time and with $O(n)$ working space, where $n$ denotes the total length of the strings in $\mathcal{S}$.
We then move on to the general case of $k > 2$ strings,
and show how to compute the set $\mathsf{M}$ of MAWs in $O(n \lceil k / \log n \rceil + |\mathsf{M}|)$ time and with $O(n (k + \log n))$ bits of working space,
in the word RAM model with machine word size $\omega = \log n$.
The latter algorithm runs in optimal $O(n + |\mathsf{M}|)$ time for $k = O(\log n)$.
\end{abstract}

\section{Introduction}

A non-empty string $w$ is said to be an \emph{absent word} (a.k.a. \emph{a forbidden word}) for a string $S$ if $w$ is \emph{not} a substring of $S$.
An absent word $w$ for $S$ is said to be a \emph{minimal absent word} (\emph{MAW})
for $S$ if all proper substrings of $w$ occur in $S$.
For instance, for string $S = \mathtt{bbacccbaa}$
over an alphabet $\Sigma = \{\mathtt{a}, \mathtt{b}, \mathtt{c}, \mathtt{d}\}$,
the set $\MAW(S)$ of all MAWs for $S$ is 
$\{\mathtt{aaa},\mathtt{bbb},\mathtt{cccc},\mathtt{d}, \mathtt{ab},\mathtt{ca},\mathtt{bc}, \mathtt{aac},\mathtt{acb},\mathtt{cbb},\mathtt{accb},$
$\mathtt{cbac}, \mathtt{bbaa}\}$.
\tmnote*{fixed example}{}%

MAWs are combinatorial string objects, and their interesting
mathematical properties have extensively been studied in the literature (see \cite{Beal1996MAWandSymbolicDynamics,Crochemore1998MAWdefinition,Fici2006MAWapplication,CrochemoreHKMPR20,MienoKAFNIBT20,DBLP:journals/tcs/AkagiKMNIBT22} and references therein).
MAWs also enjoy several applications including
phylogeny~\cite{Chairungsee2012PhylogenyByMAW},
data compression~\cite{Crochemore2000DCA,crochemore2002improved,AyadBFHP21},
musical information retrieval~\cite{CrawfordB018},
and bioinformatics~\cite{Almirantis2017MolecularBiology,Charalampapaulos2018Alignment-free,pratas2020persistent,koulouras2021significant}.

It is known that the number $|\MAW(S)|$ of MAWs for a string $S$ of length $n$
over an alphabet of size $\sigma$ is $O(\sigma n)$ and that this bound is tight~\cite{Crochemore1998MAWdefinition}.
Crochremore et al.~\cite{Crochemore1998MAWdefinition} gave an algorithm that computes
$\MAW(S)$ in $O(\sigma n)$ time with $O(n)$ working space.
Fujishige et al.~\cite{Fujishige2016DAWG} showed an improved algorithm for computing $\MAW(S)$ in optimal $O(n + |\MAW(S)|)$ time with $O(n)$ working space, for an input string $S$ of length $n$ over an integer alphabet of polynomial size in $n$.
Both of the two aforementioned algorithms utilize an $O(n)$-size string data structure called the \emph{(directed acyclic word graph)} \emph{DAWG}~\cite{BlumerBHECS85},
which recognizes the set of substrings of $S$,
and can be built in $O(n \log \sigma)$ time for general ordered alphabets~\cite{BlumerBHECS85},
and in $O(n)$ time for integer alphabets of polynomial size in $n$~\cite{Fujishige2016DAWG}.
There also exist other efficient algorithms for computing MAWs
with other string data structures such as suffix arrays and Burrows-Wheeler transforms~\cite{Belazzougui2013ESA,Barton2014MAWbySA}.

The aim of this paper is to extend the notion of MAWs to a set $\mathcal{S} = \{S_1, \ldots, S_k\}$ of multiple $k$ strings.
We are aware of a few related attempts in earlier work:
Chairungsee and Crochemore~\cite{Chairungsee2012PhylogenyByMAW} introduced a string similarity measure based on the symmetric difference $\MAW(S_1) \bigtriangleup \MAW(S_2)$ of the sets of MAWs for two strings $S_1$ and $S_2$ to compare.
They introduced a length threshold $\ell \geq 1$,
and described an approach for computing $\left( \MAW(S_1) \bigtriangleup \MAW(S_2) \right) \cap \Sigma^\ell$ with the two following steps:
First, the tries of size $O(n \ell)$ each representing the substrings of $S_1$ and $S_2$ of length up to $\ell$ are built, where $n = |S_1| + |S_2|$.
Then, two tries each representing $\MAW(S_1) \cap \Sigma^\ell$ and $\MAW(S_2) \cap \Sigma^\ell$ are built, which require $O(n \sigma)$ space.
Finally, the length-bounded symmetric difference $\left( \MAW(S_1) \bigtriangleup \MAW(S_2) \right) \cap \Sigma^\ell$
is computed from $\MAW(S_1) \cap \Sigma^\ell$ and $\MAW(S_2) \cap \Sigma^\ell$,
but the authors did not explicitly describe how this computation is done in their method.
Overall, their algorithm requires $\Omega(n (\ell + \sigma))$ time and space~\cite{Chairungsee2012PhylogenyByMAW}\footnote{The claimed time bound for computing the trie is $O(n \sigma)$ (Theorem 1 of~\cite{Chairungsee2012PhylogenyByMAW}). It seems that the authors regarded the length threshold $\ell$ as a constant.}.
Charalampapaulose et al.~\cite{Charalampapaulos2018Alignment-free}
tackled the same problem of computing the symmetric difference
$\MAW(S_1) \bigtriangleup \MAW(S_1)$ (without length threshold $\ell$),
and proposed a solution that requires $O(\sigma n)$ time and space.
Their method firstly computes $\MAW(S_1)$ and $\MAW(S_2)$ separately, and then removes the elements that are in $\MAW(S_1) \cap \MAW(S_2)$.
Charalampopoulos, Crochemore, and Pissis~\cite{charalampopoulos2018extended} presented how to count the number $|\MAW(S_1) \bigtriangleup \MAW(S_2)|$ of elements in the symmetric difference $\MAW(S_1) \bigtriangleup \MAW(S_2)$ in $O(n)$ time in the case of integer alphabets of polynomial size in $n$, by avoiding to list the elements explicitly.

Let $\mathcal{S} = \{S_1, \ldots, S_k\}$ be the input set of $k$ strings,
and $\B \in \{0,1\}^k$ be a given bit vector of length $k$.
Our problem is to list (generalized) MAWs $w$ for $\mathcal{S}$ and $\B$
such that $w \in \MAW(S_i)$
for every $\B[i] = 1$, and $w \notin \MAW(S_i)$ for every $\B[i] = 0$.
For $k = 2$, 
the aforementioned problem of computing $\MAW(S_1) \bigtriangleup \MAW(S_2)$
is equivalent to solving our problem for $\B = 01$ and $\B = 10$.
In Section~\ref{sec:overview} and Section~\ref{sec:skip_link},
we deal with the case with $k = 2$,
and present an algorithm running in $O(n+ |\mathsf{M}_{\B}|)$ time with $O(n)$ working space,
where $\mathsf{M}_\B$ denotes the set of (generalized) MAWs to output for a given bit vector $\B$ (Theorem~\ref{theo:two_strings}).
This immediately gives us an algorithm for listing the elements of
the symmetric difference $\MAW(S_1) \bigtriangleup \MAW(S_2)$
in optimal $O(n + |\MAW(S_1) \bigtriangleup \MAW(S_2)|)$ time (Corollary~\ref{coro:two_strings}).
In Section~\ref{sec:general}, we deal with the general case of $k > 2$,
and extend our solution for $k = 2$ to the general case.
Let $n$ be the total length of the input $k$ strings in $\mathcal{S}$.
Our solution for general $k > 2$ works in $O(n \lceil k / \log n \rceil + |\mathsf{M}_{\B}|)$ time
with $O(n(k+\log n))$ \emph{bits} of working space
on the word RAM model with machine word size $\omega = \log n$.
Thus, for $k = O(\log n)$,
our algorithm runs in optimal $O(n + |\mathsf{M}_{\B}|)$ time.
All the bounds claimed in this paper are valid for linearly sortable alphabets, including
integer alphabets of polynomial size in $n$.

As in the previous work~\cite{Crochemore1998MAWdefinition,Fujishige2016DAWG,Fujishige2023DAWG_MAW},
our key data structure is the DAWG for the input set $\mathcal{S}$ of strings.
The best-known algorithm for constructing the DAWG for a set of strings
of total length $n$ takes $O(n \log \sigma)$ time~\cite{Blumer1987},
thus it can require $O(n \log n)$ time for large alphabets.
We describe how the DAWG for a given set $\mathcal{S}$ of strings
over an integer alphabet of polynomial size in $n$
can be obtained in optimal $O(n)$ time (Theorem~\ref{theo:dawg_construction}), 
which may be of independent interest.

\section{Preliminaries}

\subsection{Strings}
Let $\Sigma$ be an ordered alphabet.
An element of $\Sigma$ is called a character.
For characters $a,b \in \Sigma$, we write
$a \prec b$ (or equivalently $b \succ a$) if $a$ is lexicographically smaller than $b$.
An element of $\Sigma^\ast$ is called a string.
The length of a string $S$ is denoted by $|S|$.
The empty string $\varepsilon$ is the string of length 0.
If $S = xyz$, then $x$, $y$, and $z$ are called
a \emph{prefix}, \emph{substring}, and \emph{suffix} of $S$, respectively.
They are called a \emph{proper prefix}, \emph{proper substring},
and \emph{proper suffix} of $S$ if $x \neq S$, $y \neq S$, and $z \neq S$,
respectively.
Let $\Substr(S)$ denote the set of substrings of string $S$.
For any $1 \leq i \leq |S|$, the $i$-th character of $S$ is denoted by $S[i]$.
For any $1 \leq i \leq j \leq |S|$, $S[i..j]$ denotes
the substring of $S$ starting at $i$ and ending at $j$.
For convenience, let $S[i..j] = \varepsilon$ for $0 \leq j < i \leq |S|+1$.
We say that a string $w$ \emph{occurs} in a string $S$
iff $w$ is a substring of $S$.
Note that by definition the empty string $\varepsilon$
is a substring of any string $S$
and hence $\varepsilon$ always occurs in $S$.

For a set $\mathcal{S}$ of strings,
let $\Vert \mathcal{S} \Vert$ denote the total length of the strings in $\mathcal{S}$,
that is, $\Vert \mathcal{S} \Vert = \sum_{S \in \mathcal{S}}|S|$.
Let $\Substr(\mathcal{S})$ denote the set of substrings of the strings in $\mathcal{S}$, that is, $\Substr(\mathcal{S}) = \left( \bigcup_{S \in \mathcal{S}}\{S[i..j] \mid 1 \leq i \leq j \leq |S|\} \right) \cup \{\varepsilon\}$.


%

\subsection{Minimal absent words (MAWs)}

A string $w$ is called an \emph{absent word} for a string $S$
if $w$ does not occur in $S$.
Let $\AW(S) = \Sigma^* \setminus \Substr(S)$ denote
the set of absent words for a string $S$.
An absent word $w \in \AW(S)$ for string $S$ is called a \emph{minimal absent word} or \emph{MAW} for $S$ if any proper substring of $w$ occurs in $S$.
We denote by $\MAW(S)$ the set of all MAWs for $S$.
\tmnote*{removed parentheses}{%
Let $\nonMAW(S) = \AW(S) \setminus \MAW(S)$ be the set of absent words for $S$ which are not MAWs.
}%
Note that, for strings $w$ and $S$,
it holds that $w \notin \MAW(S)$ iff $w \in \Substr(S) \cup \nonMAW(S)$.
%

We extend the aforementioned notion of MAWs to a set
$\mathcal{S} = \{S_1, \ldots, S_k\}$ of $k$ strings for $k \geq 1$, as follows:
Let $\B$ be a bit-vector of length $k$,
and let $\mathcal{S}_{\B}$ be a subset of $\mathcal{S}$
such that $\mathcal{S}_{\B} = \{S_i \mid \B[i] = 1\}$.
Let $\overline{\mathcal{S}_{\B}} = \{S_i \mid \B[i] = 0\} = \mathcal{S} \setminus \mathcal{S}_{\B}$.
A string $w$ is said to be a MAW for $\mathcal{S}_{\B}$ if
(1) $w \in \bigcap_{S_i \in \mathcal{S}_{\B}}\MAW(S_i)$ and
(2) $w \notin \bigcup_{S_i \in \overline{\mathcal{S}_{\B}}}\MAW(S_i)$.
Condition (1) implies that $w$ is a MAW for any string in $\mathcal{S}_{\B}$.
Condition (2) implies that $w$ is \emph{not} a MAW for any string in $\overline{\mathcal{S}_{\B}}$,
\sinote*{modified}{%
  which is equivalent to say that $w \in \bigcap_{S_i \in \overline{\mathcal{S}_{\B}}}(\Substr(S_i) \cup \nonMAW(S_i))$.
}%
Let $\MAW(\mathcal{S}_{\B})$ denote the set of all MAWs for $\mathcal{S}_{\B}$.

\tmnote*{fixed example}{}%
\begin{example}
  For string set $\mathcal{S} = \{\mathtt{abaab}, \mathtt{aacbba}\}$ over the alphabet $\Sigma = \{\mathtt{a, b, c, d}\}$,
  $\MAW(\mathcal{S}_{10}) = \{\mathtt{aaba}, \mathtt{bab}, \mathtt{bb}, \mathtt{c}\}$,
  $\MAW(\mathcal{S}_{01}) = \{\mathtt{ab}, \mathtt{baa}, \mathtt{bac}, \mathtt{bbb}, \mathtt{bc}, \mathtt{ca}, \mathtt{cba}, \mathtt{cc}\}$, and
  $\MAW(\mathcal{S}_{11}) = \{\mathtt{aaa, d}\}$.
\end{example}

The problem we consider in this paper is the following:

\begin{problem}[MAWs for multiple input strings]
Given a set $\mathcal{S} = \{S_1, \dots, S_k\}$ of $k$ strings
over an alphabet $\Sigma$ and a bit vector $\B$ of length $k$,
compute $\MAW(\mathcal{S}_{\B})$.
\end{problem}



\section{The DAWG data structure}

We use the \emph{directed acyclic word graph} (\emph{DAWG})~\cite{BlumerBHECS85}
data structure for a set $\mathcal{S} = \{S_1, \ldots, S_k\}$ of $k$ strings,
which is a DFA of size $O(\Vert \mathcal{S} \Vert)$
that recognizes all suffixes of the strings in $\mathcal{S}$.

To give a formal definition of $\DAWG(\mathcal{S})$,
let $\EndPos_{\mathcal{S}}(w)$ denote the set of ending positions
of all occurrences of a string $w$ in the strings of $\mathcal{S}$,
that is,
\[\EndPos_{\mathcal{S}}(w) = \{(i, j) \mid S_i[j-|w|+1..j] = w, 1 \leq i \leq k, 1 \leq j \leq |S_i|\}.
\]
We consider an equivalence relation $\equiv_{\mathcal{S}}$ of strings over $\Sigma$ w.r.t. $\mathcal{S}$ such that, for any two strings $w$ and $u$,
$w \equiv_{\mathcal{S}} u$ iff $\EndPos_{\mathcal{S}}(w) = \EndPos_{\mathcal{S}}(u)$.
For any string $x \in \Sigma^*$,
let $[x]_{\mathcal{S}}$ denote the equivalence class for $x$ w.r.t. $\equiv_{\mathcal{S}}$.
All the non-substrings $x \notin \Substr(\mathcal{S})$ form
a unique equivalence class, called the \emph{degenerate} class.

\begin{definition}
  The DAWG of a set $\mathcal{S}$ of strings, denoted $\DAWG(\mathcal{S})$,
  is an edge-labeled DAG $(V, E)$ such that
  \begin{eqnarray*}
    V & = & \{\Eqc{x}_{\mathcal{S}} \mid x \in \Substr(\mathcal{S})\}, \\
    E & = & \{(\Eqc{x}_{\mathcal{S}}, b, \Eqc{xb}_{\mathcal{S}}) \mid x, xb \in \Substr(\mathcal{S}), b \in \Sigma\}.
  \end{eqnarray*}
  We also define the set $L$ of {\em suffix links} of $\DAWG(\mathcal{S})$
  by
  \[
    L = \{(\Eqc{ax}_{\mathcal{S}}, a, \Eqc{x}_{\mathcal{S}}) \mid x, ax \in \Substr(\mathcal{S}), a \in \Sigma, \Eqc{ax}_{\mathcal{S}} \neq \Eqc{x}_{\mathcal{S}} \}.
  \]
\end{definition}
Namely, two substrings $x$ and $y$ in $\Substr(\mathcal{S})$ are
represented by the same node of $\DAWG(\mathcal{S})$
iff the ending positions of $x$ and $y$ in the strings of $\mathcal{S}$ are equal.
Note that $\DAWG(\mathcal{S})$ does not contain the node for the degenerate class
nor its in-coming edges.
This is important for $\DAWG(\mathcal{S})$ to have a total linear number
of edges~\cite{BlumerBHECS85},
and for our linear-time algorithm for listing all the MAWs for a given query.

For convenience,
assume that each string $S_i$ in $\mathcal{S}  = \{S_1, \ldots, S_k\}$
terminates with a unique end-marker $\#_i$
which does not occur elsewhere, where $\#_i \neq \#_j$ for $i \neq j$.
Then $\DAWG(\mathcal{S})$ has exactly $k$ sink nodes,
each of which recognizes all the non-empty suffixes of $S_i$.
For each $1 \leq i \leq k$,
the sink that recognizes the suffixes of $S_i$ is labeled by $i$.

The DAWG for a single string $T$ is the DAWG for a singleton $\{T\}$
and is denoted by $\DAWG(T)$.

The state-of-the-art algorithm that builds
$\DAWG(\mathcal{S})$ is Blumer et al.'s online algorithm~\cite{BlumerBHECS85}
which runs in $O(n \log \sigma)$ time with $O(n)$ space, where $n = \Vert \mathcal{S} \Vert$ is the total length of the strings in $\mathcal{S}$ and $\sigma$ is the alphabet size.
Below we describe a faster construction of $\DAWG(\mathcal{S})$ in the case of integer alphabets:
\begin{theorem}[Linear-time DAWG construction for a set of strings]
  \label{theo:dawg_construction}
  For a given set $\mathcal{S} = \{S_1, \ldots, S_k\}$ of $k$ strings of total length $n$ over an integer alphabet $\Sigma$ of polynomial size in $n$,
  one can build the edge-sorted $\DAWG(\mathcal{S})$ in $O(n)$ time and space.
\end{theorem}

\begin{proof}
  We first create a concatenated string $T = S_1 \cdots S_k$
  of total length $n$ from the strings in $\mathcal{S}$.
  We build $\DAWG(T)$ for the single string $T$ in $O(n)$ time and space, using the algorithm of Fujishige et al.~\cite{Fujishige2016DAWG,Fujishige2023DAWG_MAW},
  where the out-going edges of every node are lexicographically sorted.
  Our goal is to convert $G_T = \DAWG(T)$ to $G_{\mathcal{S}} = \DAWG(\mathcal{S})$.
  For a set $P$ of integer pairs and a pair $(a,b)$ of integers,
  let $P \oplus (a,b) = \{(p+a, q+b) \mid (p,q) \in P\}$.
  Our key observation is that, for any substrings $w \in \Substr(\mathcal{S})$
  that \emph{do not} contain separators $\#_i$ except for their last positions, it holds that
  \begin{eqnarray}
    \lefteqn{\EndPos_{\mathcal{S}}(w)} \nonumber \\
    & = & \EndPos_{S_1}(w) \cup \left( \bigcup_{2 \leq i \leq k} \EndPos_{S_i}(w) \oplus (i-1, |S_1 \cdots S_{i-1}|) \right).
    \label{eqn:end_pos}
  \end{eqnarray}
  Equation~(\ref{eqn:end_pos}) implies that the substrings $w$ of $T = S_1 \cdots S_{k}$
  which are also substrings of $\mathcal{S}$ are represented by
  essentially the same nodes in $G_T$ and in $G_{\mathcal{S}}$,
  meaning that there is an injection from the nodes of $G_{\mathcal{S}}$ to the nodes of $G_T$.

  What is left is how to remove the redundant nodes in $G_{T}$
  which represent the substrings $y$ of $T$ containing a separator $\#_i$
  inside, which are thus not substrings of $\mathcal{S}$.
  Let us call the longest path of $G_{T}$ that represents $T$ as the \emph{spine}.
  \sinote*{changed}{%
  Since each $\#_i$ occurs exactly once in $T$,
  any substrings of $T$ that contain $\#_i$ are represented by the spine of $G_T$.
  Thus, 
  we can obtain $G_{\mathcal{S}}$ by removing the redundant nodes from the spine of $G_T$, but we ensure that for every $i$ the suffixes of $S_i$ ending with $\#_i$ are still represented in the graph.
  This can be achieved as follows:
  }%
  We process $i = k, \ldots, 2$ in decreasing order.
  We first split the spine into two parts each spelling out $S_1 \cdots S_{k-1}$ and $S_k$.
  We remove the nodes in the $S_k$ part which are not reachable from the source of the modified graph, together with their out-going edges and suffix links.
  This gives us $\DAWG(\{S_1 \cdots S_{k-1}, S_{k}\})$.
  After processing $i = k$, we continue the same process for $i = k-1$
  with the remaining spine that spells out $S_1 \cdots S_{k-1}$.
  After processing $i = 2$, we obtain $G_{\mathcal{S}} = \DAWG(\mathcal{S})$.
  See Fig.~\ref{fig:DAWG_construction} for an example of our construction.
  It is trivial that all the redundant nodes can be removed in $O(n)$ time.
  \qed
\end{proof}
We remark that the order of concatenating the strings in $\mathcal{S}$ does not affect the correctness nor the complexity of our algorithm.

\begin{figure}[hp!]
  \centerline{
    \includegraphics[scale=0.52]{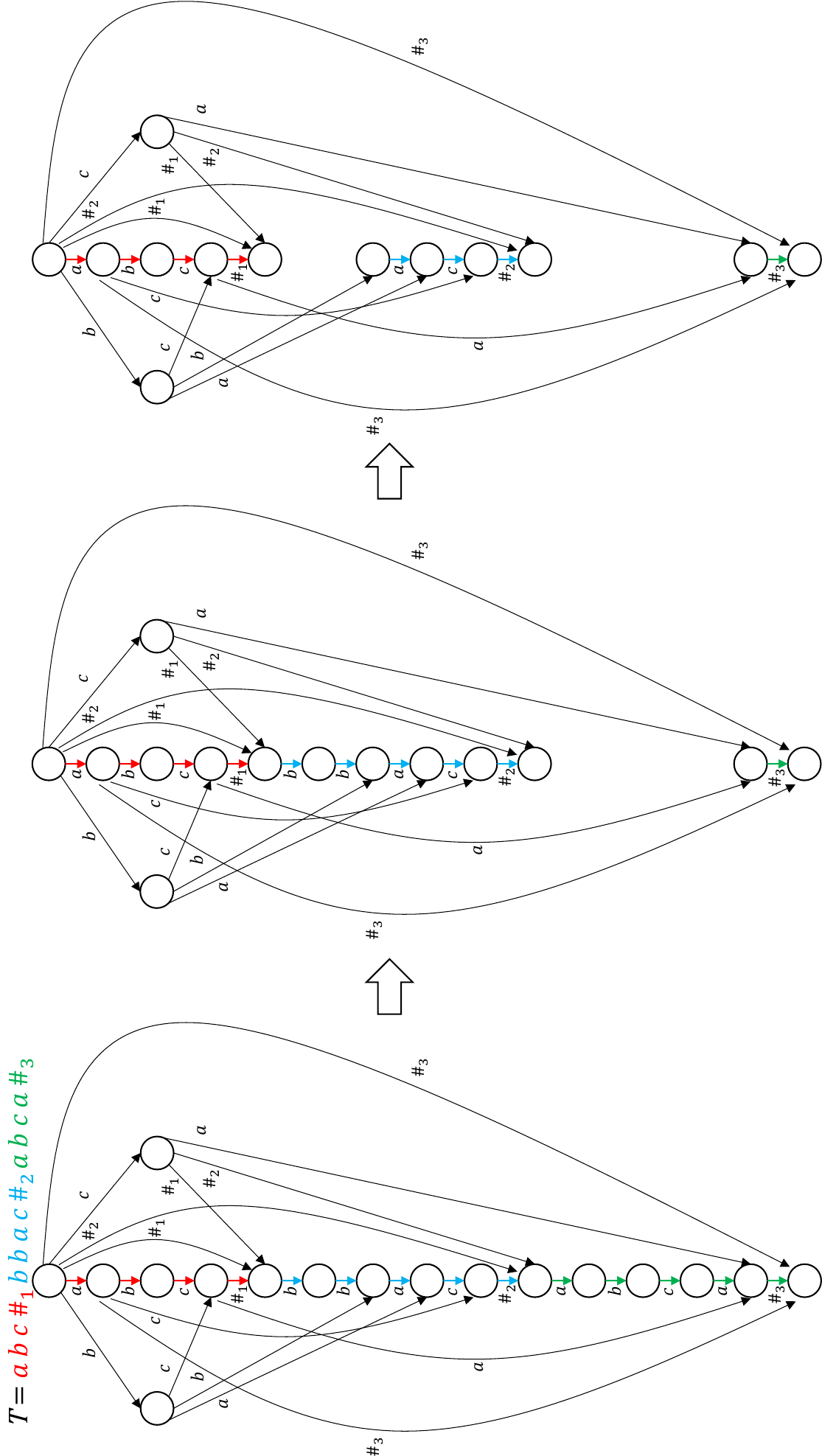}
  }
  \caption{Illustration for our linear-time construction of $\DAWG(\mathcal{S})$ for a set $\mathcal{S} = \{abc\#_1, bbac\#_2, abca\#_3\}$ of strings. We first build $\DAWG(T)$ for the concatenated string $T = abc\#_1bbac\#_2abca\#_3$. Then, we remove the redundant nodes in the spine of the DAWG for $i = 3$ and then for $i = 2$. This gives us $\DAWG(\mathcal{S})$.
  }
  \label{fig:DAWG_construction}
\end{figure}

\section{Algorithm overview for $k = 2$}
\label{sec:overview}

In what follows, we consider the case where our input set $\mathcal{S}$
consists of two strings $S_1$ and $S_2$
which respectively terminate with special characters $\#_1$ and $\#_2$.
We show how, given a bit vector $\B \in \{00, 01, 10, 11\}$ of length $2$,
we can compute $\MAW(\mathcal{S}_{\B})$ in $O(n + |\MAW(\mathcal{S}_{\B})|)$ time
and $O(n)$ working space, where $n = \Vert \mathcal{S} \Vert$.

We first build the edge-sorted $\DAWG(\mathcal{S})$ for a given $\mathcal{S} = \{S_1, S_2\}$ in $O(n)$ time and space with Theorem~\ref{theo:dawg_construction}.
We label each node $v$ of $\DAWG(\mathcal{S})$ by $\#_i$ iff $v$ represents a substring of $S_i$~($1 \leq i \leq 2$).
Let $\Label(v) \in \{\#_1, \#_2, \#_1\#_2\}$ denote the label of node $v$.
The labels of all nodes can be precomputed in $O(n)$ time.

Our algorithm is based on Fujishige et al.'s algorithm~\cite{Fujishige2016DAWG,Fujishige2023DAWG_MAW} for computing all the MAWs in the case of a single input string.
As such, for each node $x$ of $\DAWG(\mathcal{S})$
we focus on the \emph{shortest} string represented by $x$
and denote it by $au$,
where $a \in \Sigma$ and $u \in \Sigma^*$.
We use the suffix link of the node $x$ and its target node $y$
whose \emph{longest} member is $u$ (namely, the first letter $a$ of $au$ is removed by following the suffix link from $x$ to $y$).
For ease of explanation, we identify the node $x$ with the string $au$,
and the node $y$ with the string $u$.

Fujishige et al.'s algorithm compares the out-going edges of $au$
and those of $u$ one by one in the sorted order.
Suppose $au$ has an out-going edge labeled $b$.
If $u$ \emph{does not} have an out-going edge labeled $b$,
then their algorithm outputs $aub$ as a MAW for the input string.
Otherwise, it outputs nothing, and the cost is charged to
the out-going edge of $au$ labeled $b$.
Each MAW $aub$ in the output is encoded by a tuple $(a, i, j)$ such that $w[i..j] = ub$, thus taking $O(1)$ space.
This is how Fujishige et al.'s algorithm works in $O(n + |\MAW(S)|)$ time
and with $O(n)$ working space for a single string $S$.

However, in our case of multiple strings,
depending on the label of nodes $au$, $aub$ and $ub$,
and depending on the value of the given bit vector $\B$,
there may exist some edge comparisons that cannot be charged
either to the output MAWs or to the out-going edges of node $au$.
It is also possible that even if there is a node representing $aub$ in $\DAWG(\mathcal{S})$, still $aub$ is a MAW for some string(s) in $\mathcal{S}$.
To overcome these difficulties, we introduce \emph{skip links} that permit us to avoid unwanted edge character comparisons.

\section{Skip links for $k = 2$}
\label{sec:skip_link}

We use the same conventions for the nodes $au$, $aub$ and $u$ on $\DAWG(\mathcal{S})$ as in the previous section, and also consider the node $ub$.
We have three possible cases for the label of node $au$,
where $\Label(au) = \#_1\#_2$, $\Label(au) = \#_1$, or $\Label(au) = \#_2$.
In each of the three cases, there are some sub-cases for the labels of node $aub$
and node $ub$.
By inspection, we obtain
all the possible cases that need to be considered, as shown in Fig.~\ref{fig:MAW_cases}.

\begin{figure}[tb]
  \centerline{
    \begin{tabular}{|c|c|c|c|c|}
      \cline{3-5}
      \multicolumn{2}{c|}{}   & \multicolumn{3}{c|}{$au$}                                           \\
      \cline{3-5}
      \multicolumn{2}{c|}{}   & $\#_1\#_2$                & $\#_1$                    & $\#_2$      \\
      \hline
      $aub$                   & $ub$                      & \multicolumn{3}{c|}{$\B$}               \\
      \hline
      $\#_1\#_2$              & $\#_1\#_2$                & 00                        & -      & -  \\
      \hline
      \multirow{2}{*}{$\#_1$} & $\#_1$                    & 00                        & 00     & -  \\
                              & $\#_1 \#_2$               & 01                        & 00     & -  \\
      \hline
      \multirow{2}{*}{$\#_2$} & $\#_2$                    & 00                        & -      & 00 \\
                              & $\#_1 \#_2$               & 10                        & -      & 00 \\
      \hline
      \multirow{3}{*}{absent} & $\#_1$                    & 10                        & 10     & 00 \\
                              & $\#_2$                    & 01                        & 00     & 01 \\
                              & $\#_1 \#_2$               & 11                        & 10     & 01 \\
      \hline
    \end{tabular}
    \hfill
    \raisebox{-20.7mm}{
      \includegraphics[scale=0.44]{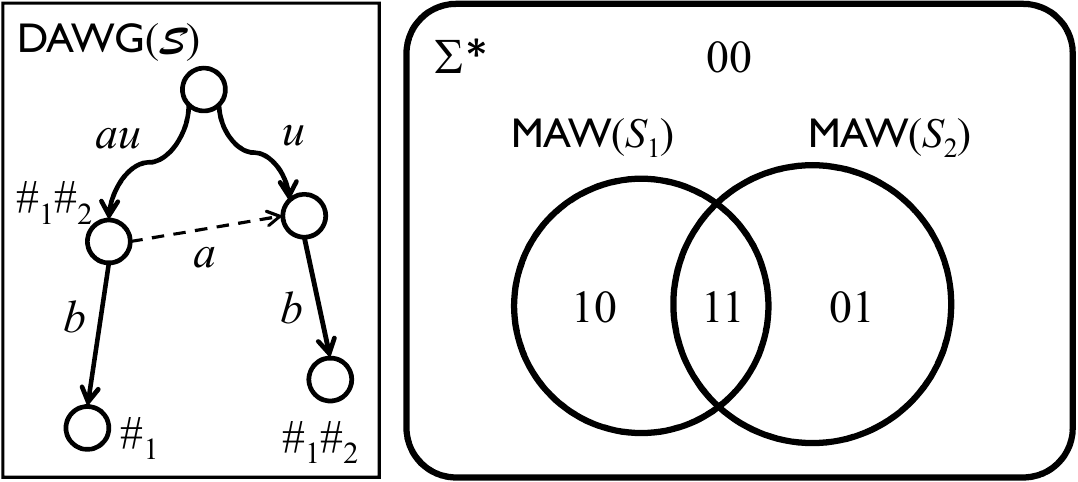}
    }
  }
  \caption{Left: All possible cases of the labels of the nodes $au$, $aub$, and $ub$, and their corresponding bit vectors $\B$. ``absent'' refers to the case where there is no out-going edge labeled $b$ from node $au$. The cells with ``-'' refer to impossible combinations of node labels. Middle: Illustration for $\DAWG(\mathcal{S})$ which shows the case where $au$ is labeled $\#_1\#_2$, $aub$ is labeled $\#_1$, and $ub$ is labeled $\#_1\#_2$. In this case $aub$ is a MAW in $\MAW(\mathcal{S}_{\B})$ with $\B = 01$ (see the left table). Right: The regions corresponding to the bit vectors $\B \in \{00, 01, 10, 11\}$.}
  \label{fig:MAW_cases}
\end{figure}

When $\B = 00$, then since $\MAW(\mathcal{S}_{00}) = \Sigma^* \setminus (\MAW(S_1) \cup \MAW(S_2))$, there are no MAWs to output.
In what follows, we describe our solutions to the cases with $\B \in \{10, 11\}$.
We remark that the case with $\B = 01$ is symmetric to the case with $\B = 10$.

\subsection{When $\B = 10$}
There are four cases in which we output $aub$ as a MAW for $\MAW(\mathcal{S}_{10})$ (see the table on the left of Fig.~\ref{fig:MAW_cases}):

\begin{enumerate}
  \item[(1)] $\Label(au) = \#_1\#_2$, $\Label(aub) = \#_2$, and $\Label(ub) = \#_1\#_2$;
  \item[(2)] $\Label(au) = \#_1\#_2$, $aub \in \AW(\mathcal{S})$, and $\Label(ub) = \#_1$;
  \item[(3)] $\Label(au) = \#_1$, $aub \in \AW(\mathcal{S})$, and $\Label(ub) = \#_1$;
  \item[(4)] $\Label(au) = \#_1$, $aub \in \AW(\mathcal{S})$, and $\Label(ub) = \#_1\#_2$.
\end{enumerate}

\subsubsection{When $\Label(au) = \#_1 \#_2$.} We create skip links that  simultaneously manage Cases (1) and (2), both having $\Label(au) = \#_1 \#_2$.
We create a selected list
$\Schar(u)$ of out-going edge labels of node $u$ such that
$\Schar(u) = \{b \mid \Label(ub) = \#_1\}$,
where the elements are lexicographically sorted.
Let $\Char(au)$ be the sorted list of all out-going edge labels of node $au$.
For any list $L$ of characters and any character $c \in \Sigma$,
let $\Succ(c, L)$ denote the lexicographical successor of $c$ in $L$.
Our algorithm for $\B = 10$ and $\Label(au) = \#_1\#_2$ is described in Algorithm~\ref{algo:B10_12}.

\begin{algorithm2e}[h]\label{algo:B10_12}
  \caption{Algorithm for $\B = 10$ and $\Label(au) = \#_1\#_2$}
  \KwIn{A node $au$ of $\DAWG(\mathcal{S})$ such that $\Label(au) = \#_1 \#_2$, $\B = 10$.}
  \KwOut{A subset $M$ of MAWs $aub$ with $b \in \Sigma$.}

  $M \leftarrow \emptyset$\;
  $U \leftarrow \Char(au) \cup \{\$_U\}$ \tcc*{$\$_U$ is lex. largest in $U$}
  $L \leftarrow \Schar(u) \cup \{\$_L\}$ \tcc*{$\$_L$ is lex. largest in $L$ and $\$_L \prec \$_u$} \label{line:schar}

  $\hat{b} \leftarrow U[1]$; $b \leftarrow L[1]$  \tcc*{start with lex. smallest characters}

  \While{$b \neq \$_L$}{
    \uIf{$\hat{b} = b$ \label{line:comparison1}}{
      \If{$\Label(aub) = \#_2$ \textbf{and} $\Label(ub) = \#_1\#_2$ \label{line:label_conditions}}{
        $M \leftarrow M \cup \{aub\}$ \tcc*{output $aub$} \label{line:output1}
      }
      $\hat{b} \leftarrow \Succ(\hat{b}, U)$ \tcc*{move to the next character in $U$}
      $b \leftarrow \Succ(b, L)$ \tcc*{move to the next character in $L$}
    }
    \uElseIf{$\hat{b} \succ b$ \label{line:comparison2}}{
      $M \leftarrow M \cup \{aub\}$ \tcc*{output $aub$} \label{line:output2}
      $b \leftarrow \Succ(b, L)$ \tcc*{move to the next character in $L$}
    }
  }

  \Return $M$\;
\end{algorithm2e}

\begin{figure}[tb]
  \centerline{
    \begin{tabular}{|c|c|c|}
      \cline{3-3}
      \multicolumn{2}{c|}{}                          & $au$                                                         \\
      \cline{3-3}
      \multicolumn{2}{c|}{}                          & $\#_1\#_2$                                                   \\
      \hline
      $aub$                                          & $ub$                              & $\B$                     \\
      \hline
      \cellcolor[gray]{0.85}$\#_1\#_2$               & \cellcolor[gray]{0.85}$\#_1\#_2$  & \cellcolor[gray]{0.85}00 \\
      \hline
      \cellcolor[gray]{0.85}{}                       & \cellcolor[gray]{0.85}$\#_1$      & \cellcolor[gray]{0.85}00 \\
      \multirow{-2}{*}{\cellcolor[gray]{0.85}$\#_1$} & \cellcolor[gray]{0.85}$\#_1 \#_2$ & \cellcolor[gray]{0.85}01 \\
      \hline
      \multirow{2}{*}{$\#_2$}                        & \cellcolor[gray]{0.85}$\#_2$      & \cellcolor[gray]{0.85}00 \\
                                                     & $\#_1 \#_2$                       & 10                       \\
      \hline
      \multirow{3}{*}{absent}                        & $\#_1$                            & 10                       \\
                                                     & \cellcolor[gray]{0.85}$\#_2$      & \cellcolor[gray]{0.85}01 \\
                                                     & \cellcolor[gray]{0.85}$\#_1 \#_2$ & \cellcolor[gray]{0.85}11 \\
      \hline
    \end{tabular}
    \hfill
    \raisebox{-22mm}{
      \includegraphics[scale=0.58]{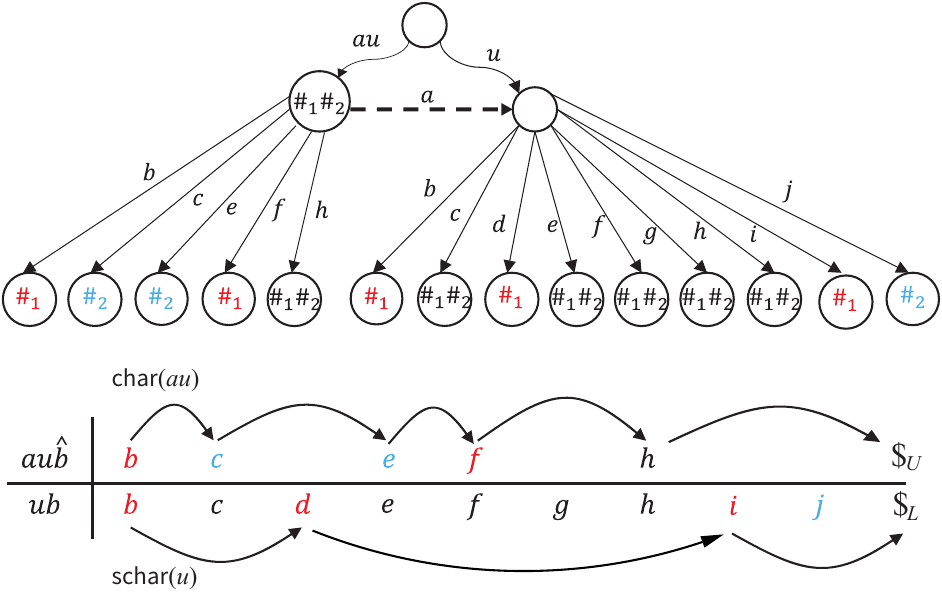}
    }
  }
  \caption{Illustration for our algorithm for $\B = 10$ and $\Label(au) = \#_1 \#_2$. The white cells in the table show the cases where we output elements of $\MAW(\mathcal{S}_{10})$. We compare the labels of the selected out-going edges of node $au$ and $u$ which are connected by the skip links, in sorted order. In this diagram, $aud$ and $aui$ are output in line~\ref{line:output2} and $auc$ and $aue$ are output in line~\ref{line:output2} of Algorithm~\ref{algo:B10_12} as elements of $\MAW(\mathcal{S}_{10})$.}
  \label{fig:MAW10_12}
\end{figure}


\subsubsection{When $\Label(au) = \#_1$.} We create skip links that simultaneously manage Cases (3) and (4), both having $\Label(au) = \#_1$.
We create another selected list
$\Schar'(u)$ of out-going edge labels of node $u$ such that
$\Schar'(u) = \{b \mid \Label(ub) \in \{\#_1, \#_1\#_2\}\}$,
where the elements are lexicographically sorted.
We use the same $\Char(au)$ in the previous case.
Our algorithm for $\B = 10$ and $\Label(au) = \#_1$ is described in Algorithm~\ref{algo:B10_1}.

\begin{algorithm2e}[h]\label{algo:B10_1}
  \caption{Algorithm for $\B = 10$ and $\Label(au) = \#_1$}
  \KwIn{A node $au$ of $\DAWG(\mathcal{S})$ such that $\Label(au) = \#_1$, $\B = 10$.}
  \KwOut{A subset $M$ of MAWs $aub$ with $b \in \Sigma$.}

  $M \leftarrow \emptyset$\;
  $U \leftarrow \Char(au) \cup \{\$_U\}$ \tcc*{$\$_U$ is lex. largest in $U$}
  $L \leftarrow \Schar'(u) \cup \{\$_L\}$ \tcc*{$\$_L$ is lex. largest in $L$ and $\$_L \prec \$_u$}

  $\hat{b} \leftarrow U[1]$; $b \leftarrow L[1]$ \tcc*{start with lex. smallest characters}

  \While{$b \neq \$_L$}{
    \uIf{$\hat{b} = b$}{
      $\hat{b} \leftarrow \Succ(\hat{b}, U)$ \tcc*{move to the next character in $U$}
      $b \leftarrow \Succ(b, L)$ \tcc*{move to the next character in $L$}
    }
    \uElseIf{$\hat{b} \succ b$}{
      $M \leftarrow M \cup \{aub\}$ \tcc*{output $aub$}
      $b \leftarrow \Succ(b, L)$ \tcc*{move to the next character in $L$}
    }
  }
  \Return $M$\;
\end{algorithm2e}

\begin{lemma}[Linear-time MAW computation for $\B = 10$]
  Given $\B = 10$,
  one can compute
  $\MAW(\mathcal{S}_{10})$ in $O(n+|\MAW(\mathcal{S}_{10})|)$ time
  and $O(n)$ working space for integer alphabets of polynomial size in $n = \Vert \mathcal{S} \Vert$.
\end{lemma}

\begin{proof}
  We run Algorithm~\ref{algo:B10_12} and Algorithm~\ref{algo:B10_1} for every node $au$ of $\DAWG(\mathcal{S})$.

  In the preprocessing phase, we build the edge-sorted $\DAWG(\mathcal{S})$ in $O(\Vert \mathcal{S} \Vert)$ time and space by Theorem~\ref{theo:dawg_construction}.
  Since the out-going edges of every node are sorted,
  we can easily compute the sorted lists $\Char(au)$, $\Schar(u)$, $\Schar'(u)$, and $\Schar''(u)$ for all nodes in $O(n)$ total time.

  Let us consider the complexity of the scanning phase of Algorithm~\ref{algo:B10_12}.
  Each edge-label comparison that falls into ``$\hat{b} = b$'' in
  line~\ref{line:comparison1} of Algorithm~\ref{algo:B10_12} is associated either to the reported MAW $aub$ if $\Label(aub) = \#_2$ and $\Label(ub) = \#_1 \#_2$ (in line~\ref{line:label_conditions} and line~\ref{line:output1}),
  or to the out-going edge of node $au$ labeled $b$ otherwise.
  Each edge-label comparison that falls into ``$\hat{b} \succ b$'' in line~\ref{line:comparison2} is associated to the reported MAW $aub$ in line~\ref{line:output2}.
  This ensures the desired time complexity for Algorithm~\ref{algo:B10_12}.
  The complexity for Algorithm~\ref{algo:B10_1} is similar to show.

  The correctness of Algorithm~\ref{algo:B10_12} and Algorithm~\ref{algo:B10_1} is immediate
  from the tables in Fig.~\ref{fig:MAW_cases} and~\ref{fig:MAW10_12}.
  %
  \qed
\end{proof}


\subsection{When $\B = 11$}
There is a single case in which we output $aub$ as a MAW for $\MAW(\mathcal{S}_{11})$ (see Fig.~\ref{fig:MAW_cases}): $\Label(au) = \#_1\#_2$, $aub \in \AW(\mathcal{S})$, and $\Label(ub) = \#_1\#_2$.

Unwanted comparisons can occur here
if $aub \in \AW(\mathcal{S})$, and $\Label(ub) = \#_1$ or $\Label(ub) = \#_2$.
To avoid such comparisons,
we consider another carefully selected list
$\Schar''(u)$ of out-going edge labels of node $u$ such that
$\Schar''(u) = \{b \mid \Label(ub) = \#_1\#_2\}$,
where the elements are lexicographically sorted.
We can use the same $\Char(au)$ in the previous subsection.

We can modify Algorithm~\ref{algo:B10_1} for $\B = 01$ with $\Label(au) = \#_1$
so that the modified algorithm computes MAWs for $\B = 11$,
only by using $\Schar''(u)$ in place of $\Schar'(u)$.
This leads us to the following lemma:

\begin{lemma}[Linear-time MAW computation for $\B = 11$]
  Given $\B = 11$,
  one can compute
  $\MAW(\mathcal{S}_{11})$ in $O(n+|\MAW(\mathcal{S}_{11})|)$ time
  and $O(n)$ working space for integer alphabets of polynomial size in $n = \Vert \mathcal{S} \Vert$.
\end{lemma}


\subsection{Our main result for $k = 2$}

Finally we obtain the main result for a case of two strings with $k = 2$.
\begin{theorem}[Linear-time MAW computation for a set of two strings]
  \label{theo:two_strings}
  Given a set $\mathcal{S} = \{S_1, S_2\}$ of two strings of total length $n$
  and a bit vector $\B \in \{01, 10, 11\}$,
  one can compute
  $\MAW(\mathcal{S}_{\B})$ in $O(n+|\MAW(\mathcal{S}_{\B})|)$ time
  and $O(n)$ working space for integer alphabets of polynomial size in $n$.
\end{theorem}

The following corollary is immediate from Theorem~\ref{theo:two_strings}.
\begin{corollary}
  \label{coro:two_strings}
  Given a set $\mathcal{S} = \{S_1, S_2\}$ of two strings of total length $n$,
  one can compute $\MAW(S_1) \cap \MAW(S_2)$,
  $\MAW(S_1) \cup \MAW(S_2)$, and $\MAW(S_1) \bigtriangleup \MAW(S_2)$
  in $O(n+|\MAW(S_1) \cap \MAW(S_2)|)$ time,
  $O(n+|\MAW(S_1) \cup \MAW(S_2)|)$ time,
  and $O(n+|\MAW(S_1) \bigtriangleup \MAW(S_2)|)$ time, respectively,
  using $O(n)$ working space, for integer alphabets of polynomial size in $n$.
\end{corollary}

\section{Algorithm for arbitrary $k > 2$}
\label{sec:general}

In this section, we present our algorithm for computing $\MAW(\mathcal{S}_{\B})$
in case where $\mathcal{S} = \{S_1, \ldots, S_k\}$ contains $k > 2$ strings.

Let $\B \in \{0,1\}^k \setminus \{0^k\}$ be an input bit vector of length $k > 2$.
We redefine the labels of the nodes of $\DAWG(\mathcal{S})$
such that $\Label(v)[i] = 1$ iff $v$ is a substring of $S_i$
for $1 \leq i \leq k$.
Namely, $\Label(v)$ is now also a bit vector of length $k$.

Let $aub \in \Sigma^*$~($a,b \in \Sigma$ and $u \in \Sigma^*$)
be a candidate of an element of $\MAW(\mathcal{S}_{\B})$
as in the previous sections, where the suffix link of node $au$ points to node $u$
and node $u$ has an out-going edge labeled $b$.
Then, it follows from the definition of $\MAW(\mathcal{S}_{\B})$ that
$aub \in \MAW(\mathcal{S}_{\B})$ iff
\begin{itemize}
  \item[(A)] $\Label(aub)[i] = 0$, $\Label(au)[i] = 1$, and $\Label(ub)[i] = 1$  (i.e. $aub \in \MAW(S_i)$),  or
  \item[(A')] $au$ has no out-going edge labeled $b$, $\Label(au)[i] = 1$, and $\Label(ub)[i] = 1$ \\ (i.e. $aub \in \MAW(S_i)$)
\end{itemize}
for all $1 \leq i \leq k$ with $\B[i] = 1$, and
\begin{itemize}
  \item[(B)] $\Label(aub)[i] = 1$ (i.e. $aub \in \Substr(S_i)$), or
  \item[(C)] $\Label(aub)[i] = 0$, and $\Label(au)[i] = 0$ or $\Label(ub)[i] = 0$ (i.e. $aub \in \nonMAW(S_i)$), or
  \item[(C')] $au$ has no out-going edge labeled $b$, and $\Label(au)[i] = 0$ or $\Label(ub)[i] = 0$ (i.e. $aub \in \nonMAW(S_i)$)
\end{itemize}
for all $1 \leq i \leq k$ with $\B[i] = 0$.

For each node $au$ in $\DAWG(\mathcal{S})$ whose suffix link points to node $u$,
we create a united single skip link $\Schar(ub)$ for the children $ub$ of node $u$
such that $b \in \Schar(ub)$ iff $\Label(ub)[i] = 1$ for every $i$ with $\B[i] = 1$.

After the above preprocessing is finished,
we proceed to the scanning phase of our algorithm.
For each node $au$, we scan the skip links $\Char(aub)$ and $\Schar(ub)$
in parallel, analogously to the case with $k = 2$.
Let $\hat{b} \in \Char(aub)$ and $b \in \Schar(ub)$.
Our algorithm compares these characters in sorted order
while keeping the invariant $\hat{b} \succeq b$ as in the case with $k = 2$.

When the comparison falls into the case ``$\hat{b} = b$'',
then we output $aub$ as an element of $\MAW(\mathcal{S}_{\B})$
if Case (A) is satisfied and if Case (B) or Case (C) is satisfied.
When the comparison falls into the case ``$\hat{b} \succ b$'',
then we output $aub$ as an element of $\MAW(\mathcal{S}_{\B})$
if Cases (A') and (C') are both satisfied.

This already gives us an $O(nk)$-time algorithm for computing $\MAW(\mathcal{S}_{\B})$
using $O(n(k+\log n))$ \emph{bits} of working space,
or alternatively $O(n \lceil k / \log n \rceil)$ words of working space
in the word RAM model with machine word size $\omega = \log n$.

We can speed up checking Cases (A), (B), (C) for each node $au$
by using bit masks of size $\omega = \log n$ each stored at nodes $aub$, $au$, and $ub$,
from $O(k)$ time to $O(\lceil k / \log n \rceil)$ time.
For Cases (A') and (C'), it suffices for us to use only the bit masks stored at nodes
$au$ and $ub$, since node $aub$ does not exist in these cases
and we detect this as a result of ``$\hat{b} \succ b$'' comparison.

Overall, we obtain the following:
\begin{theorem}[Efficient MAW computation for a set of $k$ strings]
  Given a set $\mathcal{S} = \{S_1, \ldots, S_k\}$ of $k$ strings of total length $n$ and
  a bit vector $\B \in \{0,1\}^k \setminus \{0^k\}$,  one can compute
  $\MAW(\mathcal{S}_{\B})$ in $O(n\lceil k / \log n \rceil+|\MAW(\mathcal{S}_{\B})|)$ time
  and $O(n(k+\log n))$ bits of working space (or alternatively $O(n\lceil k / \log n \rceil)$ words of working space), for integer alphabets of polynomial size in $n$.
\end{theorem}

\section{Conclusions and discussions}

In this paper, we introduced the notion of generalized MAWs for a set
$\mathcal{S} = \{S_1, \ldots, S_k\}$ of multiple $k$ strings,
and considered the problem of computing the set $\MAW(\mathcal{S}_{\B})$ of
generalized MAWs specified by a given bit vector $\B \in \{0,1\}^k$.
We presented a linear $O(n + |\MAW(\mathcal{S}_{\B})|)$-time algorithm
for $k = 2$,
and an efficient $O(n \lceil k / \log n \rceil + |\MAW(\mathcal{S}_{\B})|)$-time
algorithm for general $k > 2$.
The latter algorithm also runs in linear $O(n + |\MAW(\mathcal{S}_\B)|)$ time
for $k = O(\log n)$.
An interesting open question is whether there exists an $O(n + |\MAW(\mathcal{S}_\B)|)$-time solution for $k = \omega(\log n)$.

B{\'{e}}al et al.~\cite{BealCM03} considered a different version of MAWs $\MAW'(\mathcal{S})$ for a set $\mathcal{S}$ of $k$ strings, where a string $w = aub$ is a MAW for $\mathcal{S} = \{S_1, \ldots, S_k\}$ if $aub \notin \Substr(\mathcal{S})$, $au \in \Substr(S_i)$ and $ub \in \Substr(S_j)$ for some $1 \leq i, j \leq k$. They gave an $O(\sigma n)$-time and space solution for computing $\MAW'(\mathcal{S})$. This version of MAWs can be computed in optimal $O(n+|\MAW'(\mathcal{S})|)$ time, independently of $k$, by running our algorithm without skip links.
Ayad et al.~\cite{AyadBFHP21} considered the problem of computing the same version of MAWs of length up to $\ell$ with a given threshold $\ell > 0$.

Independently to our work, the recent work by B{\'{e}}al and Crochemore~\cite{BealC23} considered the following problem: Let $\mathcal{T}$ and $\mathcal{R}$ be sets of strings, where $\mathcal{T}$ is called a target and $\mathcal{R}$ is called a reference. A $\mathcal{T}$-specific string with respect to $\mathcal{R}$ is a string $u$ such that $u \in \Substr(\mathcal{T})$, $u \notin \Substr(\mathcal{R})$, $v \in \Substr(\mathcal{R})$ for any proper substring $v$ of $u$. By definition, a string $u$ is a $\mathcal{T}$-specific string with respect to $\mathcal{R}$ if and only if $u \in \MAW(\mathcal{R}) \cap \Substr(\mathcal{T})$.
B{\'{e}}al and Crochemore~\cite{BealC23} showed an algorithm for finding all $\mathcal{T}$-specific strings w.r.t. $\mathcal{R}$ in $O(n \sigma)$-time and $O(n)$ space, where $n$ is the total length of the strings in $\mathcal{T}$ and $\mathcal{R}$, assuming that the edges of the DAWG are represented by transition matrices (Proposition 2, \cite{BealC23}). Their algorithm also uses the DAWG built on $\mathcal{T}$ and $\mathcal{R}$ and marks its nodes in an appropriate way  (Proposition 1, \cite{BealC23}). This marking technique is very similar to our skip links from Section~\ref{sec:skip_link} for the case of $k = 2$, and thus our algorithm can be extended to solve this problem in $O(n)$ time and space for integer alphabets.

\section*{Acknowledgments}
This work was supported by JSPS KAKENHI Grant Numbers
JP23H04381~(TM),
JP21K17705, JP23H04386~(YN),
JP22H03551~(SI), JP20H04141~(HB).

\bibliographystyle{abbrv}
\bibliography{ref}

\end{document}